%% file: main.tex
\documentclass[11pt,letterpaper]{article}
\usepackage[utf8]{inputenc}
\usepackage{times}
\usepackage[ruled,vlined,commentsnumbered,titlenotnumbered]{algorithm2e}
\usepackage{url}
\usepackage{fullpage}
\usepackage{pslatex} 
\usepackage{mathdots} 
\usepackage{amsmath}
\usepackage{algorithmic}
\usepackage{amssymb}
\usepackage{amsthm}
\usepackage{color}
\usepackage{array}
\usepackage{xy}
\usepackage{setspace}
\usepackage{varwidth}
\usepackage{multicol}
\usepackage{hyperref}
\usepackage{cite}
\usepackage[margin=1in,letterpaper]{geometry}
\usepackage{accents}

\newcommand{\defn}[1]{\emph{\textbf{{#1}}}}

\usepackage{todonotes}

\newcommand{\poly}{\operatorname{poly}}
\newcommand{\polylog}{\operatorname{polylog}}
\newcommand{\E}{\mathbb{E}}

\renewcommand{\paragraph}[1]{\vspace{.5 cm} \noindent \textbf{#1} }

\usepackage{thmtools}
\usepackage{thm-restate}

\newtheoremstyle{slanted}
{3pt}
{3pt}
{\slshape}
{}
{\bfseries}
{.}
{.5em}
{}
\theoremstyle{slanted}
\newtheorem{theorem}{Theorem}
\newtheorem{lemma}[theorem]{Lemma}

\newtheorem{proposition}[theorem]{Proposition}
\newtheorem{corollary}[theorem]{Corollary}

\begin{document}

%


\date{}


\title{A Hash Table Without Hash Functions, \\and How to Get the Most Out of Your Random Bits}
\author{William Kuszmaul\footnote{MIT CSAIL, \url{kuszmaul@mit.edu}.  Funded by a Fannie and John Hertz Fellowship and an NSF GRFP Fellowship. This research was also partially sponsored by the United States Air Force Research Laboratory and the United States Air Force Artificial Intelligence Accelerator and was accomplished under Cooperative Agreement Number FA8750-19-2-1000. The views and conclusions contained in this document are those of the authors and should not be interpreted as representing the official policies, either expressed or implied, of the United States Air Force or the U.S. Government. The U.S. Government is authorized to reproduce and distribute reprints for Government purposes notwithstanding any copyright notation herein.}}
\maketitle
\thispagestyle{empty}

\begin{abstract}
This paper considers the basic question of how strong of a probabilistic guarantee can a hash table, storing $n$ $(1 + \Theta(1)) \log n$-bit key/value pairs, offer? Past work on this question has been bottlenecked by limitations of the known families of hash functions: The only hash tables to achieve failure probabilities less than $1 / 2^{\polylog n}$ require access to fully-random hash functions---if the same hash tables are implemented using the known explicit families of hash functions, their failure probabilities become $1 / \poly(n)$.

To get around these obstacles, we show how to construct a randomized data structure that has the same guarantees as a hash table, but that \emph{avoids the direct use of hash functions}. Building on this, we are able to construct a hash table using $O(n)$ random bits that achieves failure probability $1 / n^{n^{1 - \epsilon}}$ for an arbitrary positive constant $\epsilon$. 

In fact, we show that this guarantee can even be achieved by a \emph{succinct dictionary}, that is, by a dictionary that uses space within a $1 + o(1)$ factor of the information-theoretic optimum. 

Finally we also construct a succinct hash table whose probabilistic guarantees fall on a different extreme, offering a failure probability of $1 / \poly(n)$ while using only $\tilde{O}(\log n)$ random bits. This latter result matches (up to low-order terms) a guarantee previously achieved by Dietzfelbinger et al., but with increased space efficiency and with several surprising technical components.  
\end{abstract}
\vfill
\pagebreak

\newpage 
\pagenumbering{arabic}

\input{intro}

\input{rotated}

\input{succinct}

\input{appendix}

\section*{Acknowledgements.} The author is grateful to Martin Dietzfelbinger for helpful pointers to prior work, and would also like to thank Michael A. Bender and Mart\'in Farach-Colton for a number of insightful technical discussions.

\bibliographystyle{plainurl} \bibliography{writeup}

\end{document}

%% file: intro.tex
\section{Introduction}
A \defn{dictionary} is any data structure that supports insertions, deletions, and queries on a set $S$ of up to $n$ \defn{keys}; dictionaries often also allow for a user to store a value associated with each key, which can then be retrieved during queries. Unless stated otherwise, we will assume a machine word of $\Theta(\log n)$ bits, which means that keys/values are also $O(\log n)$ bits. We will also require implicitly that a dictionary should take at most linear space (i.e., $O(n \log n)$ bits) and that a dictionary should be explicit (i.e., it can be initialized in time $O(n)$). In fact, the dictionaries in this paper have the stronger property that they can be initialized in constant time. 

Randomized dictionaries are often also referred to as \defn{hash tables}.\footnote{Hash tables are sometimes also informally defined as any solution to the dictionary problem that makes use of hash functions. We intentionally take a more open-ended perspective as to the definition of a hash table, so that we include data structures that accomplish the same goal as traditionally accomplished by hash tables, but using different means.}  A hash table is said to have \defn{failure probability} $\epsilon$ if each operation takes constant time with probability at least $1 - \epsilon$, and is said to \defn{succeed with high probability} if $\epsilon \le 1 / \poly n$.

A central open question is whether there exists a deterministic constant-time dictionary. A remarkable success in this direction is P\v{a}tra\c{s}cu and Thorup's dynamic fusion node \cite{dynamicfusion}, which builds on older work by Fredman and Willard \cite{fusionoriginal} in order to construct a deterministic constant-time dictionary for very small sets of keys---that is, sets $S$ of at most $\polylog n$ keys that are $\Theta(\log n)$ bits each. For sets of $\Theta(n)$ keys, it is widely believed that (even non-explicit) deterministic constant-time dictionaries are impossible \cite{openproblems}, but we are still very far from having lower bounds to establish this (see \cite{Sundar91,HagerupMiPa01, Ruzic08,Pagh00, dietzfelbinger1994dynamic} for other related work on this question).

In this paper, we consider a natural relaxation of this question: What is the smallest failure probability that a hash table can offer \cite{iceberg, goodrich2011fully, goodrich2012cache}? We present the first hash table to achieve a significantly sub-polynomial failure probability. And we show that such a hash table can even be made \defn{succinct}, meaning that it uses space within a $(1 + o(1))$ factor of the information-theoretic optimum. 

\paragraph{Past work on super-high-probability guarantees.} 
The study of probabilistic guarantees for hash tables has, up until now, been intimately tied to the study of hash-function families \cite{naor1999construction, luby1988construct, pagh2008uniform, pagh2003uniform, kaplan2009derandomized, dietzfelbinger2003almost, dietzfelbinger1990new, Reingold1, Reingold2, patracscu2012power, siegel2004universal, siegel1989universal}. If one has access to fully-random hash functions, then it is known \cite{goodrich2011fully, goodrich2012cache, iceberg} how to achieve substantially sub-polynomial failure probabilities. However, as observed by Goodrich, Hirschberg, Mitzenmacher, and Thaler \cite{goodrich2012cache}, the known techniques for simulating constant-time hash functions with high independence \cite{siegel2004universal, pagh2008uniform, dietzfelbinger2003almost} are \emph{themselves} randomized constructions that introduce an additional $ 1/ \poly(n)$ probability of failure. Efforts at reducing these failure probabilities \cite{goodrich2012cache} have only been able to do so at the cost of $\omega(1)$ evaluation times. 

Of the known families of constant-time hash functions, the only one that has been successfully used to obtain a hash table with sub-polynomial failure probabilities is tabulation hashing \cite{patracscu2012power}. Although the standard analyses of tabulation hashing include a $1 / \poly(n)$ failure probability, it has been noted by \cite{patracscu2012power} that, in some parameter regimes, the true failure probability is actually sub-polynomial. Indeed, one can extend the techniques of \cite{patracscu2012power} to show that tabulation hash functions are load-balancing with probability $1 - 1 / 2^{\polylog n}$, thereby allowing one to construct a hash table that has a failure probability of $1 / 2^{\polylog n}$. To the best of our knowledge, this remains the smallest failure probability to be achieved by any hash table.

\paragraph{This paper: hash tables with nearly optimal failure probabilities.} 
We introduce a simple data structure, which we call the \defn{amplified rotated trie}, that offers a failure probability of $1 / n^{n^{1 - \epsilon}}$ for an arbitrarily small positive constant $\epsilon$ of our choice. Barring a deterministic constant-time dictionary, this is the close to the strongest guarantee that one could hope for: if there were to exist a hash table with failure probability $1 / n^{\epsilon n}$, for some positive constant $\epsilon > 0$, then that would imply the existence of a (non-explicit) deterministic constant-time dictionary. Our result improves significantly over the previous state-of-the-art of $1 / 2^{\polylog n}$. 

Our second result is that, with a few small modifications, the same data structure can be used to obtain a very different guarantee. The resulting hash table, which we call the \defn{budget rotated trie}, uses $\tilde{O}(\log n)$ random bits to support constant-time operations with high probability in $n$. This guarantee, which has also been achieved using more classical hashing-based techniques in previous work by Dietzfelbinger et al.~\cite{dietzfelbinger1992polynomial}, serves as a natural dual to the one above --- rather than trying to minimize failure probability, while using up to $O(n)$ random bits, one tries to minimize random bits while maintaining a standard $1 / \poly(n)$ failure probability. 

An interesting feature of budget rotated tries is that they are able to make use of so called ``gradually-increasing-independence hash functions'' \cite{Reingold1, Reingold2}. These hash functions, introduced originally by Celis, Reingold, Segev, and Wieder \cite{Reingold1} (and subsequently made more efficient by Meka, Reingold, Rothblum, and Rothblum \cite{Reingold2}) can be used to distribute $n$ balls roughly evenly across $n$ bins using only $O(n \log \log n)$ random bits, but come with the seemingly significant drawback that they require $\Theta((\log \log n)^2)$ time to evaluate. As a consequence, past work on applying these hash functions to classical hash tables \cite{Reingold3} has incurred $\omega(1)$ time per operation. Our approach suggests that such gradually-increasing-independence may be more broadly applicable to than was previously thought, and can be used in the design of constant-time data structures.

\paragraph{Achieving succinctness.}
Finally, we turn our attention to space efficiency. There has also been a great deal of work on how to construct a \defn{succinct hash table} (see, e.g., \cite{Raman03Succinct, iceberg, supersuccinct, liu2020succinct}), that is, a hash table that stores $n$ key/values pairs from a universe $U$ in space
$$(1 + o(1)) \mathcal{B}(|U|, n)$$
bits, where $\mathcal{B}(|U|, n) = \log \binom{|U|}{n}$ is the information-theoretic lower bound on the size of any hash table.

In Section \ref{sec:succinct}, we show that the data structures in this paper can also be made succinct, in the parameter regime where keys/values are $(1 + \Theta(1))\log n$ bits. More generally, we give a black-box transformation that can be applied to any dictionary in order to obtain a succinct dictionary whose probabilistic guarantees are nearly the same as the original's. The new dictionary uses $\mathcal{B}(|U|, n) + O(n (\log n) / \log \log n)$ bits.

Interestingly, the transformation itself makes use of our (non-succinct) budget rotated trie as a critical algorithmic component. The transformation also makes use of a reduction due to Raman and Rao \cite{Raman03Succinct}, and can be seen as a constant-time and randomness-efficient version of the succinct dictionary given in \cite{Raman03Succinct} (which guaranteed only constant expected-time operations).

Applying our transformation, we obtain two data structures: we get a succinct hash table that uses $O(\log n (\log \log n)^3) = \tilde{O}(\log n)$ random bits, while supporting constant-time operations with high probability; and a succinct hash table with a failure probability of $1 / n^{n^{1 - \epsilon}}$, where $\epsilon$ is an arbitrarily small positive constant of our choice. 

\paragraph{Circumventing the hash-function bottleneck. } 
At the core of our results is a simple but powerful observation: that it is possible to construct a hash table \emph{that does not use hash functions}, and that is consequently free of the limitations that hamper known hash-function constructions. In particular, we begin our exposition by constructing a simple randomized dictionary that we call a \defn{rotated radix trie}. Like standard hash tables, the rotated radix trie uses linear space and is constant-time (with high probability). But unlike standard hash tables, which rely on randomness supplied by hash functions, the rotated radix trie uses randomness directly embedded into the data structure. The rotated radix trie then serves as the basis for both the amplified rotated trie and the budget rotated trie.

\paragraph{Outline.}
The rest of the paper proceeds as follows. Section \ref{sec:prelim} presents basic preliminaries and conventions---these are specialized for the setting of non-succinct dictionaries, and slightly different conventions are established later on in Section \ref{sec:succinct} for discussing succinct dictionaries. Section \ref{sec:rotated} presents and analyzes the rotated radix trie. Building on this, Section \ref{sec:highprob} gives a hash table that achieves failure probability $1 / n^{n^{1 - \epsilon}}$ and Section \ref{sec:lowbits} gives a high-probability hash table using $O(\log n \log \log n)$ random bits; in Appendix \ref{app:hash}, we show that the latter guarantee can also be extended to the case where machine words are $\omega(\log n)$ bits. Finally, in Section \ref{sec:succinct}, we show how to transform any linear-space hash table into a succinct hash table, while nearly preserving the randomization guarantees of the data structure.

%% file: rotated.tex
\section{Preliminaries and Conventions (for Non-Succinct Dictionaries)}\label{sec:prelim}

We now present several preliminary definitions and conventions for discussing (non-succinct) dictionaries. These conventions are used throughout the paper, except in Section \ref{sec:succinct} where we consider succinct dictionaries. We end up using slightly different conventions when discussing non-succinct versus succinct dictionaries because, in the non-succinct setting, there are a number of standard simplifications that one can make without loss of generality (but which do not hold in the succinct setting). 

\paragraph{Keys, values, and dictionaries.}
Let $U = [\poly(n)]$ be the set of all possible $\Theta(\log n)$-bit keys, and let $V = [\poly(n)]$ be the set of all possible $\Theta(\log n)$-bit values. A \defn{dictionary} is a data structure that stores a set of keys from $U$, and that associates each key $x$ with a value $y \in V$. Dictionaries support three operations: Insert$(x, y)$ adds key $x$ to the set, if it is not already there, and sets the corresponding value to $y$; Delete$(x)$ removes $x$; and Query$(x)$ reports whether key $x$ is present, returning the corresponding value if so. 

When discussing non-succinct dictionaries, we focuses (without loss of generality) on fixed-capacity dictionaries, that is, dictionaries that are permitted to have up to $n$ keys at a time. Such dictionaries can be used to implement dynamically-resized dictionaries by simply rebuilding the dictionary (in a deamortized fashion) whenever its size changes by a constant factor. Unless stated otherwise, we shall require implicitly that dictionaries must use at most linear space (i.e., $O(n \log n)$ bits) and have $O(1)$ initialization time.

\paragraph{Standard techniques for simplifying dictionaries.}
There are several standard reductions that can be used to simplify the problem of maintaining a linear-space dictionary.

We can assume without loss of generality that the lifespan of a dictionary is only $O(n)$ operations. Indeed, longer sequences of operations can be broken into phases of size $ O (n) $, and the dictionary can be rebuilt from scratch during each phase (i.e., all of the elements are gradually moved from one instance of the dictionary to another new instance of the dictionary). The rebuild cost can be spread across the phase so that the asymptotic running times of operations are preserved.\footnote{For our purposes, rebuilds \emph{do not} sample new random bits. Once a dictionary's random bits are chosen, they are fixed forever.}

Since the lifespan of each phase is only $O(n)$ operations, we can implement deletions with the following trivial approach: simply mark elements as deleted, and defer the actual removal of those elements until the next rebuild. As a consequence, when designing the dictionary that will be used to implement each phase, we can assume without loss of generality that the only operations performed are insertions/queries. 

We will therefore assume throughout the paper that, whenever we are discussing a non-succinct dictionary, the sequence of operations being performed has length $O(n)$ and consists exclusively of inserts/queries. 

\paragraph{Randomization.}
Randomized dictionaries are given access to a stream of random bits---the dictionary can access the next $\Theta(\log n)$ bits of the stream in time $O(1)$. When analyzing a randomized dictionary, the goal is to bound the \defn{failure probability} for any given operation. We emphasize that, in this context, failure does not refer to lack of correctness, but instead to lack of timeliness. A dictionary \defn{fails} whenever an operation takes super-constant time. (Later, in Section \ref{sec:succinct}, when we consider succinct dictionaries, we will also allow for failures with respect to space consumption.)

All of our dictionaries share the property that, once a failure occurs, all of the rest of the operations (in the current phase of $O(n)$ inserts/queries) also fail. We will not bother to explicitly specify the dictionary's behavior when a failure occurs, since at that point it is okay for each of the remaining operations in the phase to take linear time.

We remark that randomized data structures are analyzed against \emph{oblivious} adversaries, meaning that the sequence of insertions/deletions/queries being performed is determined independently of the random bits that the dictionary uses. We also remark that the failure probability of a dictionary is determined on a per-operation basis. For example, if a dictionary has failure-probability $p$ and is used for $1 / p$ operations, then it is reasonable that some failures should occur.\footnote{Moreover, failures may be correlated between steps (and between phases). For example, if we are using $r$ random bits, and an adversary guesses them, then they can force failures all the time with probability $1 / 2^r$. }

\section{A Warmup Data Structure: The Rotated Trie.} \label{sec:rotated}

In this section, we present a simple randomized constant-time dictionary, called the \defn{rotated radix trie}, that serves as the basis for the data structures in later sections.  

\paragraph{The starting place: an $ n $-ary radix trie.}
The starting place for our data structure will be the classic $ n $-ary radix trie. Each internal node of the trie can be viewed as an array of size $ n $, where the $j$-th entry of the array stores either a pointer to child $j$, if such a child exists, or a null character otherwise. The leaves of the trie correspond to the keys in the data structure (and are where we store values). In general, there is a leaf with root-to-leaf path $j_1, j_2, j_3, \ldots, j_d$ if and only if the key $j_1 \circ j_2 \circ j_3 \circ \cdots \circ j_d \in [n^d] = [U]$ is present. 

What makes the $ n $-ary radix trie an interesting starting place is that the trie deterministically supports constant-time operations. What it does not support is space efficiency: there may be as many as $\Theta(n)$ internal nodes, each of which is an array of size $n$, and which collectively require space $\Omega(n^2)$ to implement.

\paragraph{Using randomness to save space: the rotated radix trie.}
We now add randomness to our data structure in a very simple way. Label the internal nodes of the trie by $1, 2, \ldots, m$ for some $m \in O(n)$, and refer to the array used to implement each internal node $i \in [m]$ as $A_i$. When the data structure is initialized, we assign to each internal node $ i $ a \defn{random rotation} $r_i$ selected uniformly at random from $\{0, 1, \ldots, n - 1\}$. The rotation $r_i$ is stored as part of the node $i$.

The purpose of $r_i$ is to apply a random cyclic rotation to the array $A_i$. That is, if a pointer would have been stored in position $j$ of $A_i$, it is now stored in position $((j + r_i) \bmod n)$ of $A_i$ instead. 

Finally, having rotated each of the arrays $A_i$ by $r_i$, we now overlay the arrays $A_1, A_2, \ldots, A_m$ on top of one another, and we store the contents of all of them in a single array $A$ of size $n$. Of course, the $j$-th position of $A$ may be responsible for storing elements from multiple $A_i$s. As long as the number of elements stored in each entry is relatively small, then this is fine: we simply implement each entry of $A$ as a dynamic fusion node \cite{dynamicfusion}, which is a deterministic constant-time linear-space dictionary capable of storing up to $\ell = \polylog n$ key/value pairs at a time. 

If, prior to collapsing the arrays into a single array $A$, the the $j$-th position of rotated array $A_i$ stored a pointer to array $A_{i'}$, then afterwards the dynamic fusion node $A[j]$ stores the key-value pair $(i, (i', r_{i'}))$. In this setting, we refer to the pair $(i', r_{i'})$ as a \defn{pointer} to $A_{i'}$, since it dictates which array $A_{i'}$ we are pointing at and where to find the entries of $A_{i'}$. Similarly, if prior to collapsing the arrays, the $j$-th position of the rotated array $A_i$ stored a pointer $p$ directly to a value (rather than to another array $A_{i'}$), then the dynamic fusion node $A[j]$ stores the key-value pair $(i, p)$.

\paragraph{Analyzing the rotated radix trie.}
To analyze the rotated radix trie, we must show that, with high probability in $n$, each entry of $A$ is responsible for storing entries from at most $\ell = \polylog n$ different $A_i$s. 

Let us first establish some conventions that will be useful throughout the rest of the paper. When discussing a radix trie, we will refer to the arrays $A_1, A_2, \ldots, A_m$ as the \defn{nodes} (or sometimes as the \defn{internal nodes}), and we will refer to the non-null entries of each $A_i$ (i.e., the entries containing pointers) as \defn{balls}. 

In total, there are $O(n)$ balls in the trie. Each ball $b$ is specified by a pair $(s, c) \in [m] \times [n]$, where $s \in [m]$ is the \defn{source node} for the ball (i.e., the node containing the ball), and $c \in [n]$ is the \defn{child index} of the ball (i.e., the index in $A_i$ where $b$ is logically stored). The effect of randomly rotating the arrays $A_i$ and then overlaying them to obtain a single array $A$ is that each ball $b = (s, c)$ gets mapped to position $\phi(s, c) := c + r_s$ in $A$. We refer to the entries of $A$ as \defn{bins}, so each ball $b$ gets mapped to bin $\phi(b)$. The dynamic fusion node for a each bin $j \in [n]$ stores the set of key-value pairs $(b, p)$ where $b$ ranges over the balls satisfying $\phi(b) = j$, and $p$ is the pointer corresponding to the ball $b$. 

For $i \in [m]$ and $j \in [n]$, let $X_{i, j}$ be the $0$-$1$ random variable indicating whether node $ i $ places a ball into bin $j$. The $X_{i, j}$s are not independent across the bins $j$, but they are independent across the nodes $i$, since each $X_{i, j}$ is a function of the random bits $r_i$. Therefore, the number $Y_j$ of balls in bin $j$, which is given by $Y_j =\sum_{i = 1}^m X_{i, j},$
is a sum of independent indicator random variables. 

Each of the $O(n)$ balls has probability $1/n$ of being in bin $j$, so $\E[Y_j] = O(1)$. Thus, by a Chernoff bound, we have that $Y_j \le \polylog n$ with high probability in $n$. The Chernoff bound actually tells us that $Y_j \le \polylog n$ with probability $1 / n^{\polylog n}$, so we have even achieved a slightly  sub-polynomial probability of failure.

\paragraph{Putting the pieces together.}
If we implement deletions as in Section \ref{sec:prelim}, then we obtain the following result:
\begin{proposition}
The rotated radix trie is a randomized linear-space dictionary that can store up to $n$ $\Theta(\log n)$-bit keys/values at a time, and that supports each operation in constant time with probability $1 - 1 / n^{\polylog n}$.
\label{prop:radixtrie}
\end{proposition}

It's worth taking a moment to remark on how to initialize our data structure. The random rotations $r_i$ can be initialized lazily, so that $r_i$ is generated the first time that the node $i$ is used. Additionally, we do not have to actually pay the cost of initializing any arrays, since we can use standard techniques to simulate zero-initialized arrays in constant time (see \cite{aho1974design} or Problem 9 of Section 1.6 of \cite{bentley}). Thus our rotated radix trie can be initialized in constant time.

\paragraph{Taking stock of our situation.} The rotated radix trie does not, on its own, make any significant progress on either of the problems that we care about: (1) achieving super-high probability guarantees, and (2) using a near-logarithmic number of random bits. We have achieved a \emph{slightly} sub-polynomial failure probability, but we are nowhere near our goal of $1 / n^{n^{1 - \epsilon}}$. 


What makes the rotated radix trie useful, however, is that the role of randomness in the data structure is remarkably simple. The \emph{only} sources of randomness are the rotational offsets $r_1, r_2, \ldots, r_m$. In this sense, the rotated radix trie deviates from the standard mold for how to design a constant-time dictionary. The randomness in the data structure isn't used to hash elements, but is instead used to apply random rotations to sparse arrays. 

Since the role of randomness will be important in later sections, we conclude the current section by discussing an important subtlety in how the randomness is re-purposed over time. Consider how the data structure evolves over a large period of time containing many insertions/deletions. As the shape of the trie changes, each array $A_i$ will be repurposed to represent different parts of the trie. This means that the way in which the random rotation $r_i$ interacts with the key space also changes over time, with the same $r_i$ applying to a different node in the trie (and thus a different part of the key space) at different times. The re-purposing of $r_i$s has an interesting consequence: even if two points in time $t_1$ and $t_2$ store the exact same set $S$ of key/value pairs as one-another, the shape of the rotated trie may differ considerably between the two times. 


\section{The Amplified Rotated Trie}\label{sec:highprob}

In this section, we modify the rotated radix trie to reduce its probability of failure (i.e., the probability that a given operation takes super-constant time) to $1 / n^{n^{1 - \epsilon}}$, for a positive constant $\epsilon $ of our choice. We will refer to this new data structure as the \defn{amplified rotated radix trie}. 

\paragraph{Storing overflow balls in a (non-rotated) trie.}
Whenever a ball $ b $ is inserted into a bin $ j $ that already contains $\ell = \polylog n$ other balls, the ball $b$ is regarded as an  \defn{overflow ball}. Since each bin is a dynamic fusion node with capacity $\ell$, we cannot store the overflow balls in the bins.

We instead store the overflow balls in a secondary data structure $ Q $ that is implemented as a $n^{\delta}$-ary trie, for some positive constant $\delta > 0$. 

The secondary data structure $Q$ supports inserts/queries on overflow balls in constant time. On the other hand, $Q$ is not space efficient. If there are $q$ overflow balls, then $Q$ may use as much as $q n^{\delta}$ space. To establish that our dictionary uses linear space, we must show that
\begin{equation}
\Pr[q \ge n^{1 - \delta}] \le O\left(1 / n^{n^{1 - \epsilon}}\right).
\label{eq:q}
\end{equation}

\paragraph{The problem: dependencies between balls with shared source nodes.}
Our current data structure does not yet satisfy \eqref{eq:q}, however. This is because, whenever multiple balls share the same source node, their assignments become closely linked. Suppose, for example, that the rotated trie $R$ has only $2\ell$ internal nodes, and that each internal node $i \in \{1, 2, \ldots, 2\ell\}$ contains $\Theta(n / \ell)$ balls $(i, 1), (i, 2), \ldots, (i, \Theta(n / \ell))$. With probability $1 / n^{2\ell} = 1 / 2^{\polylog n}$, each internal node $i \in \{1, 2, \ldots, 2\ell\}$ has random rotation $r_i = 0$. This results in bins $1, 2, \ldots, \Theta(n / \ell)$ each containing $2\ell$ balls---in other words, half of the balls in the system are overflow balls. This means that
$$\Pr[q \ge \Omega(n)] \ge 1 / 2^{\polylog n}.$$
That is, our failure probability using $Q$ the store overflow balls is \emph{no better} than the failure probability that we achieved in Section \ref{sec:rotated} without $Q$. 

\paragraph{Reducing the dependencies.}
What makes the above pathological example possible is that it is possible to have only a small number of internal nodes in our rotated trie. This makes it so that there are only a small number of random bits that affect the rotated trie's structure, preventing us from achieving any super-high probability guarantees. 

To fix this problem, we reduce the fanout of our rotated radix trie from $n$ to $n^{\delta}$. Now each internal node can contain at most $n^{\delta}$ balls, so there are guaranteed to be at least $n^{1 - \delta}$ internal nodes. This ensures that there are always at least $n^{1 - \delta} \log n$ random bits affecting the trie's structure.

We remark that, since the fanout of the rotated trie is now $n^{\delta}$, each ball is determined by a pair $(s, c)$ where $s \in [m]$ is a source node and $c \in [n^{\delta}]$ is a child index. Nonetheless, the mapping $\phi$ from balls to bins works exactly as before: we map ball $(s, c)$ to bin $\phi(s, c) = ((r_s + c)\bmod n)$ where $r_s \in [n]$ is selected at random. 

\paragraph{Bounding the number of overflow balls.}
Of course, there are still dependencies between the number $q_j$ of overflow balls in different bins $j \in [n]$. To handle these dependencies, we make use of a tool from probabilistic combinatorics.

Call a function $f: [0, 1)^m \rightarrow \mathbb{R}$ \defn{$L$-Lipschitz} if for every pair of inputs of the form $\vec{x} = (x_1, \ldots, x_i, \ldots, x_m)$ and $\vec{x'} = (x_1, \ldots x_i^{\prime}, \ldots, x_m)$, we have $|f(\vec{x}) - f(\vec{x'})| \le L$. McDiarmid's inequality \cite{McDiarmid89} tells us that if $f$ is $L$-Lipschitz and $X_1, X_2, \ldots, X_m \in [0, 1)$ are independent random variables, then for any $t \ge 0$,
$$\Pr[|f(X_1, \ldots, X_m) - \E[f(X_1, \ldots, X_m)]| \ge t] \le 2e^{-2t^2 / (mL^2)}.$$

To apply McDiarmid's inequality to our situation, define $f(r_1, \ldots, r_m) := q$ to be the number of overflow balls. Observe that $f$ is $n^\delta$-Lipschitz, since each $r_i$ can determine the outcome of at most $n^{\delta}$ different balls. Since $\E[q] = \frac{1}{\poly n}$, it follows by McDiarmid's inequality that
\begin{align*}
\Pr[f(r_1, \ldots, r_m) \ge n^{1 - \delta}] & \le e^{-\Omega(n^{2 - 2\delta} / (m n^{2 \delta}))} \\
& = e^{-\Omega(n^{2 - 2\delta} / n^{1 + 2\delta})} \\
& = e^{-\Omega(n^{1 - 4\delta})}. \\
\end{align*}
For any $0 < \epsilon \le 1$, we can set $\delta = \epsilon / 5$ so that 
\begin{align*}
\Pr[q \ge n^{1 - \delta}] &  \le e^{-\Omega(n^{1 - 4\delta})} \\
&\le O\left(n^{-n^{1 - \epsilon}}\right).
\end{align*}

This establishes \eqref{eq:q}. If we implement deletions as in Section \ref{sec:prelim}, then we arrive at the following theorem.

\begin{theorem}
The $n^{\epsilon / 5}$-ary amplified rotated radix trie is a randomized linear-space dictionary that can store up to $n$ $\Theta(\log n)$-bit keys/values at a time, and that supports each operation in constant time with probability $1 - O\left(1 / n^{n^{1 - \epsilon}}\right)$. 
\label{thm:high}
\end{theorem}

We remark that there is a strong sense in which the amplified rotated radix trie is nearly optimal. In particular, for any constant $\epsilon > 0$, if there were to exist a randomized linear-space dictionary with failure probability of $1 / n^{\epsilon n}$, that would imply the existence of a deterministic (though non-explicit) linear-space constant-time dictionary. 

\begin{lemma}
 Let $\epsilon > 0$ be any positive constant and assume a machine word of size $w = \Theta(\log n)$ bits. Suppose there exists randomized linear-space dictionary that stores up to $n$ $\Theta(\log n)$-bit keys/values at a time and has failure probability $1 / n^{\epsilon n}$. Then there also exists a deterministic (not-necessarily explicit) dictionary with the same guarantees.   
 \label{lem:lowerhigh}
\end{lemma}
\begin{proof}
To distinguish the randomized dictionary from the deterministic dictionary that we are constructing, we will refer to the former as a hash table and the latter as a dictionary.

As noted in Section \ref{sec:prelim}, by rebuilding our dictionary once every $O(n)$ operations, we can assume without loss of generality that the lifespan of the dictionary is at most $O(n)$ operations. We will implement the dictionary using a hash table with capacity $n' = cn$ for some large positive constant $c$ to be determined later. This means that the hash table has failure probability 
$$1 / n^{\epsilon n'} = 1 / n^{\epsilon c n}.$$

Each operation takes place on a $\Theta(\log n)$-bit key/value pair, so there are at most $n^{O(1)}$ options for what a given operation could be. The total number of $O(n)$-long operation sequences is therefore at most $n^{O(n)}$. Since our hash table has failure probability $1 / n^{\epsilon c n}$, its total failure probability on any given sequence of $O(n)$ operations is at most 
$O(n) / n^{\epsilon c n} \le 1 / n^{\epsilon c n / 2}$ The probability that there exists \emph{any} sequence of operations on which our hash table fails to be constant-time is therefore at most
$$\frac{n^{O(1)}}{n^{\epsilon c n / 2}},$$
which if $c$ is taken to be a sufficiently large constant, is at most $1/2$. Thus there exists some choice of random bits for which our hash table is constant-time on \emph{every} sequence of operations. By hard-coding in this choice of random bits, we arrive at a deterministic constant-time dictionary. 

Note that, since the hash table spends total time $O(n)$ on the $O(n)$ operations, the number of random bits that it can use is at most $O(nw ) = O(n \log n)$ bits---thus the deterministic dictionary can hard-code the random bits in linear space. 
\end{proof}

Although one typically assumes a machine-word size of $\Theta(\log n)$ bits, it is also an interesting question what the strongest achievable probabilistic guarantees are in the setting where machine words (as well as keys/values) are of some size $w = \omega(\log n)$ bits. On one hand, the larger key size makes it so that Lemma \ref{lem:lowerhigh} no longer applies, so in principle, one might be able to achieve a failure probability of $1 / n^{\omega(n)}$. On the other hand, from an upper-bound perspective, it is not even known how to achieve a \emph{sub-polynomial failure probability} in this setting \cite{goodrich2011fully, goodrich2012cache, patracscu2012power, iceberg}. Here, the main obstacle appears to be unavoidably about hash functions: can one construct a family of hash functions from $[2^w]$ to $[\poly(n)]$ such that for any given $n$-element set $S \subseteq [2^w]$, we have that $\max_{x \in S} |\{y \in S \mid h(x) = h(y)\}| \le \polylog n$ with probability $1 / n^{\omega(1)}$? If such a family were to exist, then it could be directly combined with Theorem \ref{thm:high} to construct a dictionary that achieves sub-polynomial failure probability for any key-size $w$. We conjecture that no such family of hash functions exists, and moreover, that a sub-polynomial failure probability is not possible for word sizes $w = \omega(\log n)$ bits. 

\section{The Budget Rotated Trie}\label{sec:lowbits}

In this section, we present a dictionary that uses only $ O (\log n\log\log n) $ random bits, while guaranteeing that each operation takes constant time with probability $ 1-1/\poly (n) $ (i.e., with high probability in $n$). We will refer to the data structure as the \defn{budget rotated trie}. In Appendix \ref{app:hash}, we further extend the budget rotated trie to support keys that are $\omega(\log n)$ bits, while still using only $O(\log n \log \log n)$ bits of randomness.

We remark that the guarantee achieved by the budget rotated trie is not novel---in fact, a previous approach by Dietzfelbinger, Gil, Matias, and Pippenger \cite{dietzfelbinger1992polynomial} can be used to achieve $O(\log n)$ random bits for the setting of $\Theta(\log n)$-bit keys that we are considering. Nonetheless, we believe that the construction for the budget rotated tries is interesting in its own right, both because of its relationship to the amplified rotated trie, and also because of the surprising way in which it is able to make use of gradually-increasing-independence hash functions. Additionally, the specific structure of the budget rotated trie will prove useful in our quest for succinctness in Section \ref{sec:succinct}.

Our starting place is again the rotated trie, and as in Section \ref{sec:highprob}, we will take the fanout of the trie to be $n^\delta$ for some constant $\delta$; in fact, it will suffice to simply use $\delta = 1/4$. 

\paragraph{Reducing the number of random bits to $O(n / \polylog n)$.}
To transform the $n^\delta$-ary rotated trie into a budget rotated trie, our first modification will be to reduce the number of random bits from $O(n \log n)$ to $O(n / \polylog n)$. Of course, this may not seem like much progress, but we shall see later that the distinction is important.

Recall that, in a rotated trie, each ball $b$ (i.e., each non-null entry in an internal node) contains a pointer to either a leaf (i.e., an actual key/value pair) or another internal node (i.e., a child). We now add a third option: if the ball should be pointing at another internal node $x$, but if the subtree rooted at $x$ contains fewer than $\ell = \polylog n$ total keys, then we store that subtree as a dynamic fusion node $z$. If the size of the subtree rooted at $x$ subsequently surpasses $\ell$, then we create an actual internal node for $x$---in this case, any elements stored in the fusion node $z$ remain in $z$, and the ball $b$ now stores two pointers, one to $x$ and one to $z$. In other words, there are now three possible states for a ball: it can contain a pointer to a leaf; it can contain a pointer to a dynamic fusion node; or it can contain two pointers, one to a dynamic fusion node and one to another internal node of the trie.

The point of this modification is that we only create an internal node $x$ if the subtree rooted at $x$ contains at least $\ell = \polylog n$ elements. Importantly, this means that the total number of internal nodes $m$ is at most $O(n / \ell) = n / \polylog n$. The number of random bits needed for the rotations $r_1, r_2, \ldots, r_{m}$ is therefore also $n / \polylog n$. 

\paragraph{Changing the balls-to-bins mapping.}
Our next modification is to change how we map the balls to bins. Recall that each ball $b$ is specified by a pair $(s, c)$, where $s \in [m]$ is the source node of the ball and $c \in [n^\delta]$ is the child index. In the standard rotated trie, we map balls to bins using the function
$$\phi(s, c) = \left(c + r_s\right) \bmod n.$$

We will now instead map balls to bins using the function
$$\psi(s, c) = (c + a_s (\bmod n^\delta)) \cdot n^{1 - \delta} + b_s,$$
where $a_s$ is selected at random from $[n^\delta]$ and $b_s$ is selected at random from $[n^{1 - \delta}]$.

When can think about $\psi$ as follows. We break the bins into groups $G_1, \ldots, G_{n^\delta}$ of size $n^{1 - \delta}$, and we use the random value $a_s \in [n^{\delta}]$ to assign the ball to a random group. Once the ball is assigned to a group $G_i$, it is then assigned to the $b_s$-th bin in that group. Importantly, the assignments are designed so that each source node $s$ assigns \emph{at most one} of its balls to any given group $G_i$. There will never be two balls $b_1, b_2$ in group $G_i$ that both obtain their assignments $b_s$ from the same source node. 

Since the number $ m $ of internal nodes may be as large as  $n / \polylog n$, we cannot afford to generate $a_1, a_2, \ldots, a_m \in [n^{\delta}]$ and $b_1, b_2, \ldots, b_m \in [n^{1 - \delta}]$ truly at random. Fortunately, as we shall now see, the roles of the $a_i$s and $b_i$s have been carefully designed so that both sequences can be generated using a small number of ``seed'' random bits.

\paragraph{Generating the $a_i$s with $O(1)$-independent hash functions.} Let $k$ be a sufficiently large positive constant, and select a random hash function $g:[n] \rightarrow [n^\delta]$ from a family of $k$-independent hash functions. Since $k = O(1)$, the function $g$ can be specified using $O(\log n)$ random bits, and can be evaluated in time $O(1)$. We compute the $a_i$s by
$$a_i := g(i).$$  

To analyze the number of balls in each group $G_i$, we use a well-known tail bound for $k$-independent random variables (see, e.g., \cite{bellare1994randomness} or \cite{dubhashi2009concentration}).

\begin{lemma}[Lemma 2.2 of \cite{bellare1994randomness}]
Let $k \ge 4$ be an even integer. Suppose $X_1, \ldots, X_m$ are $k$-wise independent 0-1 random variables. Let $X = \sum_i X_i$. Then, for any $t \ge 0$,
$$\Pr[|X - \E[X]| \ge t] \le 2 \left(\frac{nk}{t^2}\right)^{k / 2}.$$
\label{lem:alpha}
\end{lemma}

Define $X_j$ to be the event that source-node $j$ sends a ball to group $G_i$. The $X_j$s are $k$-independent, so we have by Lemma \ref{lem:alpha} that
$$\Pr[|G_i| - \E[|G_i|] \ge n^{0.75}] \le 2 \left(\frac{kn}{n^{1.5}}\right)^{k / 2} \le n^{-\Omega(k)} =  1 / \poly(n).$$ 
Since $\delta = 0.25$, it follows that
$$\Pr[|G_i| - \E[|G_i|] \ge n^{1 - \delta}] \le 1 / \poly(n).$$
Since each of the $O(n)$ balls is equally likely to be in any group, we have that $\E[|G_i|] = O(n^{1 - \delta})$. Thus
$$\Pr[|G_i| \le O(n^{1 - \delta})] \ge 1 -  1 / \poly(n).$$
That is, each group $G_i$ contains at most $O(n^{1 - \delta})$ balls with high probability in $ n $.

\paragraph{Generating the $b_i$s with increasing-independence hash functions.}
To generate the $ b_i $s without using a large number of random bits, we make use of a more sophisticated family of hash functions. Call a family $\mathcal{H}(t)$ of hash functions $h:[\poly(t)] \rightarrow [t]$ \defn{load-balancing} if it can be used to map $t$ balls to $t$ bins with maximum load $\polylog t$; that is, for any fixed set $S \subseteq \poly(t)$ of size $t$, and for any fixed $i \in [t]$, if we select a random $h \in \mathcal{H}$, then  $$|\{s \in S \mid h(s) = i\}| \le \polylog t$$ with probability $1 - 1 / \poly(t)$.

Celis, Reingold, Segev, and Wieder \cite{Reingold1} showed how to construct a load-balancing family $\mathcal{H}(t)$ of hash functions such that each $h \in \mathcal{H}$ can be described with $O(\log t \log \log t)$ random bits and can be evaluated in time $O(\log t \log \log t)$. The family $\mathcal{H}$ is referred to as having ``gradually-increasing-indepenence'' because each $h \in \mathcal{H}$ is actually the composition of $\Theta(\log \log t)$ hash functions $h_1, \ldots, h_{\Theta(\log \log t)}$ with different levels of independence: each $h_i$ determines $\Theta((3/4)^i \log t)$ bits of $h$, and each $h_i$ is $(1 / \poly t)$-close to being $\Theta((4/3)^i)$-independent. 

The family $\mathcal{H}$ comes with a tradeoff. It is able to achieve a maximum load of $\polylog t$ (in fact, it even achieves maximum load $O(\log t / \log \log t)$) using on $O(\log t \log \log t)$ bits, but it requires super-constant time to evaluate. Subsequent work \cite{Reingold2} has improved the evaluation time from $O(\log t \log \log t)$ to $O((\log \log t)^2)$. It seems unlikely that the evaluation time can be improved to $O(1)$, however, since $\Omega(\log \log t)$ time is needed just to read the random bits used to evaluate the hash function.

The super-constant evaluation time makes it so that hash functions with gradually-increasing independence are not suitable for direct use in constant-time hash tables \cite{Reingold3}. We get around this problem by using $h$ not as a hash function but as a pseudo-random number generator. Specifically, we select a random $h: [m] \rightarrow [n^{1 - \delta}]$ from $\mathcal{H}(n^{1 - \delta})$, and we use $h$ to initialize the $b_i$s as
$$b_i := h(i).$$

Since $h$ takes time $O((\log \log n)^2)$ to evaluate, each $b_i$ now takes time $O((\log \log n)^2)$ to initialize. Recall, however, that we only create a new internal node $x$ in our rotated trie once there are more than $\ell = \polylog n$ records that want to reside in that node's subtree; the first $\ell = \polylog n$ insertions that wish to use $x$ are instead placed into a dynamic fusion node that acts as a proxy for $x$. As a result, we can afford to spend up to $\ell$ time \emph{initializing} the node $x$, and we can spread that time across the $\ell$ insertions that trigger $x$'s initialization. Since $\ell = \omega((\log \log n)^2)$, we can initialize $b_i = h(i)$ without any problem.

\paragraph{Analyzing the maximum load.}
Recall that, with probability $1 - 1 / \poly(n)$, each group $G_i$ contains at most $O(n^{1 - \delta})$ balls. Furthermore, each of the balls have different source nodes than one another. If a ball has source-node $s$, then it is placed in the $b_s$-th bin of $G_i$. 

Let $S_i \subseteq [m]$ be the set of source nodes that assign balls to $G_i$. Then for each $r \in [n^{1 - \delta}]$, the number $g_{i, r}$ of balls in the $r$-th bin of $G_i$ is given by
$$g_{i, r} = |\{s \in S_i \mid h(s) = r\}|.$$
Since $h: \poly(n) \rightarrow [n^{1 - \delta}]$ is from a load-balancing family of hash functions, we are guaranteed to have
$$g_{i, r} \le \polylog n^{1 - \delta} \le \polylog n$$ 
with high probability in $n$. 

\paragraph{Putting the pieces together.}
The fact that each bin contains at most $\polylog n$ balls (with high probability) means that, as in the standard rotated trie, each bin can be implemented with a dynamic fusion node. Operations on our dictionary therefore take time $O(1)$ with high probability in $n$. If we implement deletions as in Section \ref{sec:prelim}, then we arrive at the following theorem.

\begin{theorem}
The budget rotated trie is a randomized linear-space dictionary that can store up to $n$ $\Theta(\log n)$-bit keys/values at a time, that uses $O(\log n \log \log n)$ random bits, and that supports each operation in constant time with probability $1 - 1 / \poly(n)$.
\label{thm:budget}
\end{theorem}


We conclude the section by observing that there is a strong sense in which the guarantee achieved by the budget rotated trie is optimal. In particular, if there were to exist a hash table  failure probability $1 / n^c$ but that used fewer than $c \log n$ random bits, then there would also necessarily exist a deterministic linear-space constant-time dictionary. 

\begin{lemma}
Suppose there exists a randomized linear-space dictionary that can store up to $n$ $\Theta(\log n)$-bit keys/values at a time, that uses $c \log n$ random bits, but that has a failure probability smaller than $1 / n^c$. Then there exists a deterministic dictionary with the same guarantees.
\end{lemma}
\begin{proof}
To distinguish the randomized dictionary from the deterministic dictionary that we are constructing, we will refer to the former as a hash table. Let $R$ denote the $c \log n$ random bits used by the hash table. Define $\mathcal{D}$ to be the deterministic dictionary obtained by setting $R = 0$. Suppose for contradiction that $D$ is not constant time. Then there exists some sequence of operations such that the final operation on $D$ takes super-constant time. This means that, with probability at least $1 / n^c$, that same operation would have taken super-constant time in our hash table. But the hash table has failure probability smaller than $1 / n^c$, a contradiction. 
\end{proof}

%% file: succinct.tex
\section{Achieving Succinctness}\label{sec:succinct}

In this section, we turn our attention to space efficiency. Throughout the section, we will use $c_1 \log n$ to denote the size in bits of each key, we will use $c_2 \log n$ to denote the size in bits of each value, and will assume that $c = c_1 + c_2$ is a positive constant larger than $1$. Here, unlike in previous sections, we use $n$ to denote the \emph{current size}, at any given moment, and we allow $n$ to change dynamically over time, so long as the key/value length remains $\Theta(\log n)$ bits at all times after initialization---this means that the constants $c_1, c_2, c$ also change dynamically.\footnote{Note that keys trivially must have length at least $\log n$ bits, so the key/value size is necessarily $\Omega(\log n)$. If we assume that values are asymptotically no larger than keys, then one can enforce the bound of keys/values having length $O(\log n)$ by treating the dictionary as containing an extra $2^{\epsilon w}$ dummy elements, where $w$ is the key/value length and $\epsilon$ is a small positive constant.} 

The dictionaries that we have described in previous sections are already optimal up to constant factors, using a total of $\Theta(n \log n)$ bits to store $n$ key/value pairs. We shall now strive to achieve optimal space consumption up to low-order terms, that is, to use a total of
\begin{equation}
(1 + o(1)) \log \binom{2^{c \log n}}{n} = cn\log n - n \log n + o(n \log n)
\label{eq:optimalspace}
\end{equation}
bits. Such a dictionary is referred to as \defn{succinct} \cite{Raman03Succinct}.

In fact, we will prove a much more general result: that any constant-time dictionary can be \emph{transformed} into a succinct constant-time dictionary, while (nearly) preserving the random-bit usage and failure probability of the original dictionary. 

Define an \defn{$(r(n), p(n), s(n))$-dictionary} to be any constant-time dictionary that uses $O(r(n))$ random bits, achieves failure probability $O(p(n))$ per operation, and uses space $cn\log n - n \log n + O(s(n))$ bits. Since we are interested in dictionaries that automatically resize as the number of keys change, one should think of $r(n)$, $p(n)$, and $s(n)$ as functions rather than fixed values. Note that, in the context of space-efficient time-efficient dictionaries, a failure event could be in terms of either time (an operation takes $\omega(1)$ time) or space (the dictionary fails to fit in $cn\log n - n \log n + O(s(n))$ bits)---and a failure probability of $p(n)$ means that, at any given moment, the probability of a failure occurring should be at most $O(p(n))$. 

The main result of the section can be stated as follows.

\begin{theorem}
Let $\epsilon$ be a small positive constant. Suppose that $r(n)$ and $p(n)$ are nondecreasing functions satisfying $r(n) \le O(n)$ and $\exp(-n^{1 - \epsilon}) \le p(n) \le 1 / \polylog(n)$. 
 Given an $(r(n), p(n), n \log n)$-dictionary, one can construct a $(r'(n), p'(n), s'(n))$-dictionary with 
$$r'(n) = r(n) + (\log p(n)^{-1}) \cdot (\log \log n)^3,$$
$$p'(n) = p(n / \log \log n),$$
and
$$s'(n) = \frac{n \log n}{\log \log n}.$$
\label{thm:transform}
\end{theorem}

Applying Theorem \ref{thm:transform} to Theorems \ref{thm:high} and \ref{thm:budget}, we get the following succinct versions of the theorems. 

\begin{corollary}
Let $0 < \epsilon < 1$ be a positive constant. There exists a $(r(n), p(n), s(n))$-dictionary that uses $r(n) = O(n)$ random bits, that incurs a failure probability $p(n) = \exp(-n^{1 - \epsilon})$, and that incurs an additive space overhead of $s(n) = O\left(\frac{n \log n}{\log \log n}\right)$ bits compared to the information-theoretical optimum.
\label{cor:high}
\end{corollary}

\begin{corollary}
There exists a $(r(n), p(n), s(n))$-dictionary that uses $r(n) = O(\log n (\log \log n)^3) = \tilde{O}(\log n)$ random bits, that incurs a failure probability $p(n) = 1 / \poly(n)$, and that incurs an additive space overhead of $s(n) = O\left(\frac{n \log n}{\log \log n}\right)$ bits compared to the information-theoretical optimum.
\label{cor:budget}
\end{corollary}

The rest of the section will be spent proving Theorem \ref{thm:transform}. We begin in Subsection \ref{sec:reduction} by presenting a reduction due to Raman and Rao \cite{Raman03Succinct}, which transforms the problem of constructing a succinct dictionary into a different problem, which we call the \defn{many-sets problem}. Then, in Subsection \ref{sec:ourreduction}, we show how to solve the many-sets problem using an $(r(n), p(n), O(n\log n))$-dictionary, while approximately preserving the randomness guarantees of the dictionary---an interesting feature of our solution is that it makes extensive use of the budget rotated trie constructed in Section \ref{sec:lowbits}.

\subsection{Reduction to the Many-Sets Problem}\label{sec:reduction}

An important tool in our proof of Theorem \ref{thm:transform} will be a reduction due to Raman and Rao \cite{Raman03Succinct}. This reduction transforms the problem of constructing a succinct dynamic dictionary into a different problem that we call the many-sets problem.

Let $\delta > 0$ be a small positive constant of our choice. The \defn{many-sets problem}, with parameter $\delta$, is defined as follows. Let $S_1, S_2, \ldots, S_{m}$ be (dynamically changing) sets of $O(\log n)$-bit key/value pairs, satisfying $\sum_i |S_i| = n$, $m \le n / \polylog(n)$, and $|S_i| \le n^\delta$ for all $i \in [m]$, where $n$ and $m$ are permitted to evolve dynamically over time. We will use $\gamma_i = O(\log n)$ to denote the size of each key/value pair in $S_i$ (so different $S_i$s may have differently sized key/value pairs). 

Any solution to the \defn{many-sets problem} must support the following operations in constant time: 
\begin{itemize}
    \item \textsc{Insert$(i, x, y)$}, which inserts key/value pair $(x, y)$ into $S_i$; \item \textsc{Delete$(i, x)$}, which deletes $x$ (and its corresponding value) from $S_i$; \item  and \textsc{Query$(i, x)$}, which returns the value associated with $x$ in $S_i$, or declares that $x \not\in S_i$.\end{itemize} A many-sets solution is said to be an \defn{$(r(n), p(n), s(n))$-solution} if it uses $r(n)$ random bits, has failure probability $p(n)$ per insertion, and uses total space
$$\sum_i |S_i| \cdot \left(\gamma_i + O(\log |S_i|)\right) + s(n)$$
bits.

The following reduction is given (implicitly) in Section 3 of \cite{Raman03Succinct}:
\begin{theorem}
The problem of constructing a $(r(n), p(n), s(n) + n \log n / \log \log n)$-dictionary reduces deterministically to the problem of constructing an $(r(n), p(n), s(n))$-solution to the many-sets problem, where the parameter $\delta$ is a positive constant of our choice.
\label{thm:RR}
\end{theorem}

\paragraph{Understanding the reduction.} Although we defer the full proof of Theorem \ref{thm:RR} to \cite{Raman03Succinct}, we take a moment to briefly describe the high-level structure here. The key insight is to make use of tries in a clever way (that differs substantially from how they are used elsewhere in this paper). Different nodes of the trie have different fanouts: if a node has $r > \polylog n$ keys in it and has depth $O(1)$ in the trie, then it has fanout $r / \polylog n$. If a node has either $r \le \polylog n$ keys in it, or has sufficiently large constant depth in the trie, then the node is a \defn{leaf}, and the elements of the node correspond to a set $S_i$ in the many-sets problem. 

Unlike for the data structures in this paper, the internal nodes of this trie are implemented using straightforward arrays: if a node has fanout $f$, it is implemented using an array of size $f$. Fortunately, by choosing the fanout $f$ to be $r / \polylog n$, where $r$ is the number of elements stored by the node, one ensures that the arrays used to implement internal nodes have cumulative size $n / \polylog n$. 

A key insight is that the trie allows for us to shave bits off of keys: if a node $x$ has fanout $f$, then we can remove the high-order $\log f$ bits from each of the keys stored in $x$'s children (these bits are now stored implicitly in the trie path). Raman and Rao \cite{Raman03Succinct} show that, if a leaf has size $|S_i|$, then the length $\gamma_i$ of each of the key/value pairs in $S_i$ will satisfy $\gamma_i \le c \log n  - \log n - q \log |S_i| + O(\log \log n)$ bits for a positive constant $q$ of our choice (depending on the maximum depth of the trie). This is why, for the many-sets problem, it is okay to use (amortized) space $\gamma_i + O(\log |S_i|)$ bits per key in $S_i$. 

We remark that, in \cite{Raman03Succinct}, Theorem \ref{thm:RR} was proved with \emph{amortized} constant-time operations. The amortization came from the fact that, whenever a node $x$'s size changes substantially in the trie, it must be rebuilt; however, by spreading the rebuild across $\Theta(|S_i|)$ operations (that each modify the node $x$), the rebuild can trivially be deamortized to take $O(1)$ time per operation (and without compromising space efficiency, since each element is in only one version of the node at a time).

Before continuing, it is worth taking a moment to understand where the bounds $m \le n / \polylog(n)$ and $|S_i| \le n^\delta$ (which are assumed in the many-sets problem) come from. The bound $m \le n / \polylog(n)$ comes from the fact that each internal node of the trie has fanout $r / \polylog(n)$, where $r$ is its size; since the trie has depth $O(1)$, this means that the total number of leaves (and thus the number of $S_i$s) is at most $O(n / \polylog n)$. The bound $|S_i| \le n^\delta$ comes from the fact that, if a given $S_i$ were to have size $> n^\delta$, then at each node on its root-to-leaf path, each element in the set would have $\delta \log n - O(\log \log n)$ bits shaved off; but assuming that the trie has sufficiently large constant depth, this means that \emph{all} of the bits are shaved off from the elements of $S_i$, which is a contradiction.

\paragraph{A starting place: Raman and Rao's solution to the many-sets problem.}
Raman and Rao \cite{Raman03Succinct} give a simple solution to the many-sets problem that incurs constant expected time per operation, and which will serve also as the starting point for our construction. At a high level, each $S_i$ is stored in a two-part structure, consisting of a \defn{skeleton} $A_i$ and a \defn{storage array} $B_i$.\footnote{The names \emph{skeleton} and \emph{storage array} are conventions that we are establishing here for ease of discussion, and were not used in the original paper \cite{Raman03Succinct}.}

The storage array $B_i$ is logically implemented as a dynamically-resized array that stores the elements of $S_i$ (both the keys and values) contiguously. The dynamic resizing of the $\{B_i\}$s can be implemented (and deamortized) to take $O(1)$ time per operation while ensuring that, in aggregate, the $B_i$ arrays use space within a factor of $1 + 1 / \polylog(n)$ of optimal \cite{Raman03Succinct} (see, also, \cite{brodnik1999resizable}, for a more detailed treatment of succinct dynamic arrays). 

The skeleton $A_i$ allows for queries to $S_i$ to find the appropriate key/value in $B_i$. In particular, the skeleton is a dictionary that maps the $\Theta(\log |S_i|)$-bit hash $h(x)$ for each key $x \in S_i$ to the index $j \in B_i$ at which $x$ appears.\footnote{As noted by by \cite{Raman03Succinct}, there is a subtle issue that one must be careful about for deletions. Whenever a deletion occurs, a gap is created in some position $j$ of  $B_i$. To remove this gap, one must move the final element $x$ in $B_i$ to position $j$. This means that we must also update the index that is stored for element $x$ in $A_i$ to be $j$. }

The good news is that, within the skeleton $A_i$, both the keys (i.e., hashes $h(x)$) and the values (i.e., indices in $B_i$) are only $O(\log |S_i|)$ bits---this means that $A_i$ can be implemented using any standard linear-space dictionary, without worrying about succinctness. The bad news is that there are two ways in which an insertion into $A_i$ could fail: (1) the hash $h(x)$ of the element being inserted satisfies $h(x) = h(x')$ for some other $x' \in S_i$; or (2) the dictionary used to implement $ A_i $ fails. 

Assuming that $A_i$ is implemented to have a $1 / \poly(|S_i|)$ failure probability (i.e., to be a w.h.p.~dictionary), the probability of any given insertion into $S_i$ failing is $1 / \poly(|S_i|)$. Of course, each $|S_i|$ can be arbitrarily small, which is why the data structure given in \cite{Raman03Succinct} offered constant expected-time operations, rather than a high-probability guarantee.

\subsection{Proof of Theorem \ref{thm:transform}}\label{sec:ourreduction}

By the reduction in the previous subsection (Theorem \ref{thm:RR}), it suffices for us to prove the following proposition:

\begin{proposition}
Let $\epsilon$ be a small positive constant. Suppose that $r(n)$ and $p(n)$ are nondecreasing functions satisfying $r(n) \le O(n)$ and $\exp(-n^{1 - \epsilon}) \le p(n) \le 1 / \polylog(n)$. 
 Given an $(r(n), p(n), n \log n)$-dictionary, one can construct a $(r'(n), p'(n), s'(n))$-solution to the many-sets problem, with parameter $\delta = 2\epsilon$, and with 
$$r'(n) = r(n) + (\log p(n)^{-1}) \cdot (\log \log n)^3,$$
$$p'(n) = p(n / \log \log n),$$
and
$$s'(n) = \frac{n \log n}{\log \log n}.$$
\label{prop:transform}
\end{proposition}

We will prove Proposition \ref{prop:transform} by adapting the skeleton/storage-array approach outlined in the previous section. The basic idea will be to implement each $A_i$ as a budget rotated trie, and then to share random bits between $A_i$s in just the right way so that (1) the total number of random bits used a small; but (2) with good probability, only an $O(1 / \log \log n)$ fraction of the elements present at any given moment will have experienced a failure when they were inserted. We then store the elements that experience failures in a backyard data structure implemented using the $(r(n), p(n), O(n \log n))$-dictionary that we are given.

\paragraph{Preliminaries on how to discuss the $A_i$s.}
Since the size of a given skeleton $A_i$ may fluctuate over time, the skeleton $A_i$ will need to be rebuilt every time that its size changes by a constant factor. These rebuilds can be deamortized to take constant time per operation (and since the dictionary $A_i$ is permitted to use $O(|S_i| \log |S_i|)$ bits, it  is acceptable for the rebuilds to add a constant-factor space overhead to $A_i$). 

At any given moment, define the \defn{skeletal size} $a_i$ of $S_i$ to be the size that $S_i$ was the most recent time that $A_i$ was (logically) rebuilt. The skeletal size satisfies $a_i = \Theta(|S_i|)$ at all times, but has the convenient property that it only changes when a rebuild occurs. The skeletal size $a_i$ will also be used to determine how we implement $A_i$ (sets with different skeletal sizes may be implemented differently from one another)---this means that, each time $A_i$ is rebuilt (and $a_i$ changes), the implementation of  $A_i$ may also change. 

\paragraph{Defining a backyard structure.} An important component of the data structure will be a \defn{backyard dictionary} $T$, which is used to store a small number of keys/values space inefficiently. We implement $T$ using the $(r(n), p(n), n\log n)$-dictionary that we are given, and we assume without loss of generality that $|T| \ge n / \log \log n$ at all times (if $|T| < n / \log \log n$, then we can pad it with dummy elements to bring the size up). This means that, at any given moment, $T$ uses at most $O(r(n))$ random bits, has failure probability at most $O(p(n / \log \log n))$, and uses space at most
$$O\left((|T| + n / \log \log n) \cdot \log n\right)$$
bits. In order for $T$ to meet the constraints of Proposition \ref{prop:transform} we will need to establish that, with probability $1 - O(p'(n))$, we have 
\begin{equation}
|T| \le O(n / \log \log n)
\label{eq:Tsmall}
\end{equation}
at any given moment.

It is worth taking a moment to describe precisely what information $T$ stores for each key/value pair $(x, y)$ that it stores belonging to a given skeleton $A_i$: it stores the pair $(i, x)$ as a key and it stores $y$ as a value. Since different $A_i$s have different key/value sizes (all of which are $\Theta(\log n)$), we pad the sizes of the keys/values to all be a fixed length $\Theta(\log n)$ in the backyard.\footnote{As an additional subtlety, whenever an element is stored in the backyard, it should also be stored in the appropriate storage array $B_i$, that way whenever $A_i$ is rebuilt, it can determine which of its elements are currently in the backyard, it can remove those elements from the backyard, and it can include those elements in the rebuild of $A_i$.}

\paragraph{Storing small $A_i$s in the backyard automatically.}
Let $c$ be a sufficiently large positive constant. We handle the $A_i$s satisfying $a_i \le \log^{c} n$ by simply placing \emph{all} of their elements in the backyard $T$. The number of such $A_i$s is trivially at most $m$, which by the definition of the many-sets problem, is at most $n / \polylog n \le O(n / \log^{2c} n)$. The number of elements that these $A_i$s contribute to the backyard is therefore at most $O(\log^{c} n \cdot n / \log^{2c} n) = o(n / \log^c n)$.  Throughout the rest of the section, we will focus exclusively on $A_i$s satisfying $a_i > \log^{c} n$. 


\paragraph{Partitioning the remaining $A_i$s into groups $G_{j, k}$.}
We say that each $A_i$ is in \defn{category} $j = \lfloor \log a_i \rfloor$. Within each category $j$, we partition the $A_i$s (satisfying $a_i > \log^c n$) into 
\begin{equation}
    t_j = \frac{(\log \log n)^2 \log p(n)^{-1} }{j}
    \label{eq:deft}
\end{equation}
groups, $G_{j, 1}, G_{j, 2}, \ldots, G_{j, t_j}$ such that, for each group $k \in [t_j]$,
$$\sum_{A_i \in G_{j, k}} a_i \le O(n / t_j).$$
Note that such a partition is feasible only if and only if we can guarantee that $a_i \le O(n / t_j)$ for each $i$---fortunately, this follows from 
\begin{align*}
a_i \le n^{\delta} & \le O\left(n \cdot \frac{(\log \log n)^2}{j n^{1 - \epsilon}}\right) \text{   \phantom{fffffffffffff}    (since }\epsilon = \delta / 2\text{)} \\ & \le O\left(n \cdot \frac{(\log \log n)^2 \log p(n)^{-1}}{j }\right) \text{ (since }\exp(-n^{1 - \epsilon}) \le p(n) \text{)}\\
&= O(n / t_j).
\end{align*}
For any given category $j$, the only difference between how we implement each of the groups $G_{j, 1}, G_{j, 2}, \ldots, G_{j, t_j}$ is that each group $G_{j, k}$ is implemented using a different sequence $R_{j, k}$ of $\Theta(j \log j)$ random bits. 

\paragraph{Implementing each $A_i$.} We now describe how to implement a given skeleton $A_i \in G_{j, k}$ using the $\Theta(j \log j) = \Theta(\log a_j \log \log a_j) = \Theta(\log |A_j| \log \log |A_j|)$ random bits $R_{j, k}$. Recall that we store the $A_i$s satisfying $a_i \le \log^{c} n$ in the backyard, so we need only focus here on $A_i$s satisfying $a_i > \log^c n$. 

Define a hash function $h_{j, k}$ mapping the elements $x \in A_i$ to $\Theta(\log |A_i|)$-bit string $h(x)$. The hash function is implemented using the family of hash functions given in Lemma \ref{lem:hash} of Appendix \ref{app:hash}, meaning that the hash function uses $\Theta(\log |A_j|)$ random bits, and avoids collisions on $A_i$ with probability $1 - 1 / \poly(|A_i|)$ (note that, since $a_i > \log^{c} n$, the precondition for the lemma is met). The skeleton $A_i$ will map hashes $h_{j, k}(x)$, for $x \in A_i$, to indices in $B_i$.

We implement the dictionary $A_i$ using a budget rotated trie (Theorem \ref{thm:budget})---this uses $\Theta(\log |A_j| \log \log |A_j|)$ random bits and has failure probability $1 / \poly(|A_j|)$ per insertion. We remark that, in addition to making use of the probabilistic guarantees offered by the budget rotated trie, we will also be making use of the especially simple way in which the data structure experiences failures: the only possible failure mode is that one of the bins in the trie overflows. This will allow for us to gracefully handle when the $A_j$s fail: we store the element that experienced failure in a backup data structure, and we allow the $A_j$ that caused the failure to continue as though that insertion never happened. 

In more detail, there are two reasons that an insertion into $A_i$ might fail: (1) the hash $h_{j, k}(x)$ of the element being inserted satisfies $h_{j, k}(x) = h_{j, k}(x')$ for some other $x' \in S_i$; or (2) the budget rotated trie used to implement $ A_i $ fails (i.e., the insertion would cause one of the bins in the trie to overflow). The probability of either of these events occurring at any given insertion is $1 / \poly(|A_i|) = 1 / \poly(2^j)$.  Whenever an insertion into $A_i$ fails, we store the key/value pair in the backyard data structure $T$ instead. This means that our full data structure has two possible failure modes: the case where $T$ itself fails, and the case where $T$ becomes too large, violating \eqref{eq:Tsmall}.

\paragraph{Analyzing the probability of a failure.}
The backyard data structure $T$ has failure probability at most $O(p(n / \log \log n))$, by design. Thus, our task is to bound the probability that \eqref{eq:Tsmall} fails. 

\begin{lemma}
With probability $1 - \poly(p(n))$, the category $j$ contributes at most $O(n / (\log \log n)^2)$ elements to $T$, at any given moment.
\label{lem:smallbackyard}
\end{lemma}
\begin{proof}
We have already established that the $A_i$s satisfying $a_i \le \log^c n$ deterministically contribute at most $O(n / (\log \log n)^2)$ elements to $T$ (because they deterministically \emph{have} at most $O(n / (\log \log n)^2)$ elements). Thus we focus here on the $A_i$s satisfying $a_i > \log^c n$ (note that these are all in categories $j$ satisfying $j > \Omega(\log \log n)$.  

Now consider a category $j$ satisfying $j = \Omega(\log \log n)$. As shorthand, we will use $t$ to denote $t_j$ and $G_k$ to denote $G_{j, k}$. For each group $G_{k}$, $k \in [t]$, define $X_{k}$ to be the number of elements that skeletons in $G_{k}$ contribute to the backyard. We have that $X_{k} \le O(n / t)$ deterministically; and, since each element $x$ in each $A_i \in G_{k}$ has a $1 / \poly(2^j)$ probability of failure, we have that
\begin{equation}
\E[X_{k}] = \frac{n}{t \poly(2^j)}.
\label{eq:EXk}
\end{equation}
Finally, since each of the groups $G_{1}, G_{2}, \ldots, G_{t}$ use different random-bit sequences, the $X_{k}$s are independent. Defining $X'_{k} = X_{k} / \Theta(n / t)$, the number of elements that $G_{k}$ contributes to the backyard can be expressed as $\Theta(\frac{n}{t} \cdot Y)$, where
$$Y = \sum_{k = 1}^{t} X'_{k}.$$
The $X'_{k}$'s are independent random variables in the range $[0, 1]$ and, by \eqref{eq:EXk},
\begin{equation}
    \E[Y] = \sum_{k = 1}^{t} \E[X'_{k}] = \sum_{k = 1}^t  \Theta(t / n) \E[X_{k}] = t / \poly(2^j).
    \label{eq:EY}
\end{equation}
Applying a Chernoff bound to $Y$, we get
\begin{align*}
\Pr[Y > (1 + D)\E[Y]] \le \left( \frac{e^D}{(1 + D)^{1 + D}}\right)^{\E[Y]},
\end{align*}
which for $D \ge \Omega(1)$ implies
\begin{align*}
\Pr[Y_j > D \E[Y]] \le D^{-\Omega(D \cdot \E[Y])}.
\end{align*}
Set
\begin{align*}
    D &= \frac{t}{(\log \log n)^2 \E[Y]} \\ & = \frac{\poly(2^j)}{(\log \log n)^2 } \text{\phantom{ffffffffffffffff} (by \eqref{eq:EY})} \\ & \ge 2^{\alpha j} \text{\phantom{fffffffffffffff} \phantom{fffffffffff}  (since }j\ge \Omega(\log \log n)\text{)},
\end{align*}
where $\alpha$ is a positive constant of our choice. Then we have
\begin{align*}
\Pr\left[Y_j > \frac{t}{(\log \log n)^2}\right] & \le  D^{-\Omega(D \cdot \E[Y])} \\ 
& = \exp\left(- \alpha j \cdot \Omega(D \cdot \E[Y]) \right) \text{\phantom{fffffffffff} (since } D \ge 2^{\alpha j} \text{)}\\
& = \exp\left(- \alpha j \cdot \Omega(t / (\log \log n)^2) \right) \\
& = \exp\left(- \alpha j \cdot \Omega(\log p(n)^{-1} / j)\right) \text{\phantom{ffff} (by \eqref{eq:deft})}\\
& = \exp\left(- \alpha \cdot \Omega(\log p(n)^{-1})\right) \\
& = \poly(p(n)).
\end{align*}
Thus, we have with probability $1 - \poly(p(n))$ that
$$Y \le \frac{t}{(\log \log n)^2}.$$
The number of elements that category $j$ contributes to the backyard is therefore at most
$$O\left(\frac{n}{t} \cdot Y\right) \le O(n / (\log \log n)^2),$$
as desired.
\end{proof}

Putting the pieces together, the total size of $T$ is $O(n / \log \log n)$ with probability
$$1 - O(p(n / \log \log n)) - (\log \log n) \cdot \poly(p(n)) = 1 - O(p(n / \log \log n)).$$

\paragraph{Bounding the number of random bits.}
Finally, we count the number of random bits used by the data structure. The backyard uses $O(r(n))$ random bits, so it suffices to bound the number of random bits used by the skeletons in each category. Note that different categories can use the same random bits as one another (since we do not require independence between categories), so it suffices to bound the number of random bits used by any given category of skeletons. In category $j$, there are $t_j$ groups, each of which uses $\Theta(j \log j) = O(j \log \log n)$ random bits. The total number of random bits used by the category is therefore
$$O(j (\log \log n) t_j) = O\left(j \log \log n \cdot \frac{(\log \log n)^2 \log p(n)^{-1} }{j} \right) = O\left((\log \log n)^3 \log p(n)^{-1}\right).$$
In total, the number of random bits used by the data structure is
$$O(r(n)) + O\left((\log \log n)^3 \log p(n)^{-1}\right).$$
This completes the proof of Proposition \ref{prop:transform}, and thus also the proof of Theorem \ref{thm:transform}.

%% file: appendix.tex
\appendix
\section{Appendix: Universe Reduction Using $ O (\log n) $ Random Bits}\label{app:hash}

In this section, we extend the budget rotated trie to support keys from a universe $ U $ of super-polynomial size. Throughout the section, we set $U = [2^u]$ for some $u = n^{o(1)}$, and we assume that machine words are $\Theta(u)$ bits. 

To support large keys, the natural approach is to first hash elements from $U$ to a smaller universe $U'$ of polynomial size, an then to store the $\Theta(\log n)$-bit keys in a hash table along with pointers to the full keys/values. Past work on load-balancing hash functions \cite{Reingold1} has used a pair-wise independent hash function $h: [2^u] \rightarrow [\poly(n)]$ to perform this reduction. This requires the use of $\Theta(u)$ random bits, which when $u$ is large, is significantly larger than $\log n \log \log n$. 

An appealing alternative to using pairwise-independent hash functions would be to instead use Pagh's construction \cite{pagh2009dispersing} (which, in turn, is based on an earlier construction by Fredman, Koml\'os, and Semer\'edi \cite{Fredman82FKS}) of $(1 + o(1))$-universal hash functions that require only $O(\log n + \log \log u)$ random bits. The only minor problem with this construction is that it is not fully explicit. The construction requires access to a random prime number $p \in [\poly(n)]$, but the only known time-efficient high-probability approaches to constructing such a prime number require $\omega(\log n \log \log n)$ random bits (see discussion in \cite{fouque2014close}). 


Fortunately, this issue is relatively straightforward to solve. For completeness, we now give a construction for a simple family of hash functions that can be initialized in time $o(n)$ and used for universe reduction.

\begin{lemma}
Let $n > u^c$ for a sufficiently large positive constant $c$ and let $S \subseteq [2^u]$ be a set of size $n$. Let $\mathcal{P}$ be the set of prime numbers in the range $[n^{2/c}]$. Select $p_1, p_2, \ldots, p_{c^2}$ independently and uniformly at random from $\mathcal{P}$, and define the function $h:[2^u] \rightarrow [n^{2c}]$ by 
$$h(x) = (x \bmod p_1 p_2 \cdots p_{c^2}).$$
With probability $1 - 1 /\poly(n)$, $h$ is injective on $S$.
\label{lem:hash}
\end{lemma}
\begin{proof}
The probability that $|h(S)| \neq S$ satisfies
\begin{align*}
\Pr[|h(S)| \neq S] & \le \sum_{s_1, s_2 \in S} \Pr[|s_1 - s_2| \text{ divisible by all of }p_1,p_2, \ldots, p_{c^2}],\\
\end{align*}
where $s_1$ and $s_2$ are implicitly taken to be distinct. Since the $p_i$s are independent, this is
\begin{align*}
& \sum_{s_1, s_2 \in S} \left(\Pr[|s_1 - s_2| \text{ divisible by }p_1]\right)^{c^2}. 
\end{align*}
The quantity $|s_1 - s_2|$ is an element of $U = [2^u]$, and can thus have at most $u$ distinct prime factors. Therefore,
\begin{align*}
\Pr[|h(S)| \neq S] & \le \sum_{s_1, s_2 \in S} \left(\frac{u}{|\mathcal{P}|}\right)^{c^2}\\
& \le \sum_{s_1, s_2 \in S} \left(\frac{n^{1/c}}{|\mathcal{P}|}\right)^{c^2}.\\
\end{align*}
By the Prime Number Theorem, the set $\mathcal{P}$ of primes in the range $[n^{2/c}]$ has size $\Omega(n^{2/c} / \log n)$. Therefore,
\begin{align*}
\Pr[|h(S)| \neq S] & \le \sum_{s_1, s_2 \in S} O\left(\frac{n^{1/c}}{n^{2/c} / \log n}\right)^{c^2}\\
& \le O\left(\sum_{s_1, s_2 \in S} \left(\frac{\log n}{n^{1/c}}\right)^{c^2}\right)\\
& \le O\left(\sum_{s_1, s_2 \in S} \frac{\log^{c^2} n}{n^{c}}\right) \\
& \le O\left(\frac{n^2 \log^{c^2} n}{n^{c}}\right) \\
& \le 1 / \poly(n).
\end{align*}
\end{proof}

Since all of the prime numbers in $[n^{\epsilon}]$ can be enumerated in time $O(n^{2\epsilon})$, we get the following corollary:

\begin{corollary}
Let $u = n^{o(1)}$. For any constant $\delta > 0$, there exists an explicit family $\mathcal{H}$ of constant-time hash functions $h:[2^u] \rightarrow [\poly(n)]$ such that (a) a random function $h \in \mathcal{H}$ can be constructed in time $O(n^{\delta})$ using $O(\log n)$ random bits; and (b) for any fixed set $S \subseteq U$ of size $n$, and for a random $h \in \mathcal{H}$, we have that $|h(U)| = |U|$ with probability $1 - 1 / \poly(n)$. 
\label{cor:hash}
\end{corollary}

We can use Corollary \ref{cor:hash} to construct a version of the budget rotated trie that supports large universes.

\begin{theorem}
Let $u = n^{o(1)}$, suppose that keys/values are $u$ bits, and assume a machine word of size at least $\Omega(u)$ bits. The budget rotated trie uses $O(\log n \log \log n)$ random bits, it uses $O(n u)$ bits of space, and it supports insert/delete/query operations on up to $n$ keys/values at a time. The data structure can be initialized in time $O(n^{\epsilon})$, for a positive constant $\epsilon$ of our choice, and each insert/delete/query operation takes constant time with probability $1 - 1 / \poly(n)$.
\label{thm:budget2}
\end{theorem}

To eliminate the $O(n^\epsilon)$ initialization cost, we can also construct a dynamic version of the same data structure, where there is some upper bound $N$ on the data structure's size, but where the true size $n$ changes over time. Every time that the data structure's size changes by a constant factor, we rebuild it based on the new value of $ n $. Each rebuild takes time $O(n)$ (with high probability in $n$), but the cost of a rebuild can be spread across $\Theta(n)$ operations. The properties of this new data structure can be summarized with the following corollary.

\begin{corollary}
Let $u = N^{o(1)}$, suppose that keys/values are $u$ bits, and assume a machine word of size at least $\Omega(u)$ bits. The dynamic budget rotated trie uses $O(\log N \log \log N)$ random bits and supports insert/delete/query operations on up to $N$ keys/values at a time. If it is storing $n$ key/value pairs, then it uses $O(n u)$ bits of space, and each insert/delete/query operation takes constant time with probability $1 - 1 / \poly(n)$.
\end{corollary}

%% file: main.bbl
\begin{thebibliography}{10}

\bibitem{aho1974design}
Alfred~V Aho and John~E Hopcroft.
\newblock {\em The design and analysis of computer algorithms}.
\newblock Pearson Education India, 1974.

\bibitem{bellare1994randomness}
Mihir Bellare and John Rompel.
\newblock Randomness-efficient oblivious sampling.
\newblock In {\em Proceedings 35th Annual Symposium on Foundations of Computer
  Science (FOCS)}, pages 276--287. IEEE, 1994.

\bibitem{iceberg}
Michael~A Bender, Alex Conway, Mart{\'\i}n Farach-Colton, William Kuszmaul, and
  Guido Tagliavini.
\newblock All-purpose hashing.
\newblock {\em arXiv preprint arXiv:2109.04548}, 2021.

\bibitem{supersuccinct}
Michael~A. Bender, Martín Farach-Colton, John Kuszmaul, William Kuszmaul, and
  Mingmou Liu.
\newblock On the optimal time/space tradeoff for hash tables.
\newblock In {\em Proceedings of the Fifty-Fourth Annual ACM Symposium on
  Theory of Computing (STOC)}, 2022.

\bibitem{bentley}
Jon Bentley.
\newblock {\em Programming pearls}.
\newblock Addison-Wesley Professional, 2016.

\bibitem{brodnik1999resizable}
Andrej Brodnik, Svante Carlsson, Erik~D Demaine, J~Ian Ian~Munro, and Robert
  Sedgewick.
\newblock Resizable arrays in optimal time and space.
\newblock In {\em Workshop on Algorithms and Data Structures}, pages 37--48.
  Springer, 1999.

\bibitem{Reingold1}
L~Elisa Celis, Omer Reingold, Gil Segev, and Udi Wieder.
\newblock Balls and bins: Smaller hash families and faster evaluation.
\newblock In {\em Proceedings of the 2011 IEEE 52nd Annual Symposium on
  Foundations of Computer Science (FOCS)}, pages 599--608, 2011.

\bibitem{dietzfelbinger1994dynamic}
M~Dietzfelbinger, A~Karlin, K~Mehlhorn, FM~auf~der Heide, H~Rohnert, and
  RE~Tarjan.
\newblock Dynamic perfect hashing: upper and lower bounds.
\newblock In {\em Proceedings of the 29th Annual Symposium on Foundations of
  Computer Science (FOCS)}, pages 524--531, 1988.

\bibitem{dietzfelbinger1990new}
Martin Dietzfelbinger and Friedhelm~Meyer auf~der Heide.
\newblock A new universal class of hash functions and dynamic hashing in real
  time.
\newblock In {\em International Colloquium on Automata, Languages, and
  Programming (ICALP)}, pages 6--19. Springer, 1990.

\bibitem{dietzfelbinger1992polynomial}
Martin Dietzfelbinger, Joseph Gil, Yossi Matias, and Nicholas Pippenger.
\newblock Polynomial hash functions are reliable.
\newblock In {\em International Colloquium on Automata, Languages, and
  Programming}, pages 235--246. Springer, 1992.

\bibitem{dietzfelbinger2003almost}
Martin Dietzfelbinger and Philipp Woelfel.
\newblock Almost random graphs with simple hash functions.
\newblock In {\em Proceedings of the Thirty-Fifth Annual ACM Symposium on
  Theory of Computing (STOC)}, pages 629--638, 2003.

\bibitem{dubhashi2009concentration}
Devdatt~P Dubhashi and Alessandro Panconesi.
\newblock {\em Concentration of measure for the analysis of randomized
  algorithms}.
\newblock Cambridge University Press, 2009.

\bibitem{fouque2014close}
Pierre-Alain Fouque and Mehdi Tibouchi.
\newblock Close to uniform prime number generation with fewer random bits.
\newblock In {\em International Colloquium on Automata, Languages, and
  Programming (ICALP)}, pages 991--1002. Springer, 2014.

\bibitem{Fredman82FKS}
Michael~L. Fredman, Janos Komlos, and Endre Szemeredi.
\newblock Storing a sparse table with o(1) worst case access time.
\newblock In {\em Proceedings of the 23rd Annual Symposium on Foundations of
  Computer Science (FOCS)}, pages 165--169. IEEE Computer Society, 1982.

\bibitem{fusionoriginal}
Michael~L Fredman and Dan~E Willard.
\newblock Blasting through the information theoretic barrier with fusion trees.
\newblock In {\em Proceedings of the twenty-second annual ACM Symposium on
  Theory of Computing (STOC)}, pages 1--7, 1990.

\bibitem{goodrich2011fully}
Michael~T Goodrich, Daniel~S Hirschberg, Michael Mitzenmacher, and Justin
  Thaler.
\newblock Fully de-amortized cuckoo hashing for cache-oblivious dictionaries
  and multimaps.
\newblock {\em arXiv preprint arXiv:1107.4378}, 2011.

\bibitem{goodrich2012cache}
Michael~T Goodrich, Daniel~S Hirschberg, Michael Mitzenmacher, and Justin
  Thaler.
\newblock Cache-oblivious dictionaries and multimaps with negligible failure
  probability.
\newblock In {\em Mediterranean Conference on Algorithms}, pages 203--218.
  Springer, 2012.

\bibitem{HagerupMiPa01}
Torben Hagerup, Peter~Bro Miltersen, and Rasmus Pagh.
\newblock Deterministic dictionaries.
\newblock {\em Journal of Algorithms}, 41(1):69--85, 2001.
\newblock URL:
  \url{https://www.sciencedirect.com/science/article/pii/S019667740191171X},
  \href {https://doi.org/https://doi.org/10.1006/jagm.2001.1171}
  {\path{doi:https://doi.org/10.1006/jagm.2001.1171}}.

\bibitem{kaplan2009derandomized}
Eyal Kaplan, Moni Naor, and Omer Reingold.
\newblock Derandomized constructions of k-wise (almost) independent
  permutations.
\newblock {\em Algorithmica}, 55(1):113--133, 2009.

\bibitem{liu2020succinct}
Mingmou Liu, Yitong Yin, and Huacheng Yu.
\newblock Succinct filters for sets of unknown sizes.
\newblock In {\em 47th International Colloquium on Automata, Languages, and
  Programming (ICALP)}. Schloss Dagstuhl-Leibniz-Zentrum f{\"u}r Informatik,
  2020.

\bibitem{luby1988construct}
Michael Luby and Charles Rackoff.
\newblock How to construct pseudorandom permutations from pseudorandom
  functions.
\newblock {\em SIAM Journal on Computing}, 17(2):373--386, 1988.

\bibitem{McDiarmid89}
Colin McDiarmid.
\newblock On the method of bounded differences.
\newblock {\em Surveys in combinatorics}, 141(1):148--188, 1989.

\bibitem{Reingold2}
Raghu Meka, Omer Reingold, Guy~N Rothblum, and Ron~D Rothblum.
\newblock Fast pseudorandomness for independence and load balancing.
\newblock In {\em International Colloquium on Automata, Languages, and
  Programming (ICALP)}, pages 859--870. Springer, 2014.

\bibitem{naor1999construction}
Moni Naor and Omer Reingold.
\newblock On the construction of pseudorandom permutations: {Luby---Rackoff}
  revisited.
\newblock {\em Journal of Cryptology}, 12(1):29--66, 1999.

\bibitem{pagh2003uniform}
Anna Ostlin and Rasmus Pagh.
\newblock Uniform hashing in constant time and linear space.
\newblock In {\em Proceedings of the thirty-fifth annual ACM symposium on
  Theory of Computing (STOC)}, pages 622--628, 2003.

\bibitem{pagh2008uniform}
Anna Pagh and Rasmus Pagh.
\newblock Uniform hashing in constant time and optimal space.
\newblock {\em SIAM Journal on Computing}, 38(1):85--96, 2008.

\bibitem{Pagh00}
Rasmus Pagh.
\newblock Faster deterministic dictionaries.
\newblock In {\em Proceedings of the Eleventh Annual ACM-SIAM Symposium on
  Discrete Algorithms (SODA)}, pages 487--493, USA, 2000.

\bibitem{pagh2009dispersing}
Rasmus Pagh.
\newblock Dispersing hash functions.
\newblock {\em Random Structures \& Algorithms}, 35(1):70--82, 2009.

\bibitem{patracscu2012power}
Mihai P{\v a}tra{\c{s}}cu and Mikkel Thorup.
\newblock The power of simple tabulation hashing.
\newblock {\em Journal of the ACM (JACM)}, 59(3):1--50, 2012.

\bibitem{dynamicfusion}
Mihai Patrascu and Mikkel Thorup.
\newblock Dynamic integer sets with optimal rank, select, and predecessor
  search.
\newblock In {\em 2014 IEEE 55th Annual Symposium on Foundations of Computer
  Science (FOCS)}, pages 166--175, 2014.

\bibitem{Raman03Succinct}
Rajeev Raman and Satti~Srinivasa Rao.
\newblock Succinct dynamic dictionaries and trees.
\newblock In {\em Automata, Languages and Programming}, pages 357--368, Berlin,
  Heidelberg, 2003. Springer Berlin Heidelberg.

\bibitem{Reingold3}
Omer Reingold, Ron~D Rothblum, and Udi Wieder.
\newblock Pseudorandom graphs in data structures.
\newblock In {\em International Colloquium on Automata, Languages, and
  Programming (ICALP)}, pages 943--954. Springer, 2014.

\bibitem{Ruzic08}
Milan Ru\v{z}i\'{c}.
\newblock Uniform deterministic dictionaries.
\newblock {\em ACM Trans. Algorithms}, 4(1), March 2008.
\newblock \href {https://doi.org/10.1145/1328911.1328912}
  {\path{doi:10.1145/1328911.1328912}}.

\bibitem{siegel1989universal}
A~Siegel.
\newblock On universal classes of fast high performance hash functions, their
  time-space tradeoff, and their applications.
\newblock In {\em Proceedings of the 30th Annual Symposium on Foundations of
  Computer Science (FOCS)}, pages 20--25, 1989.

\bibitem{siegel2004universal}
Alan Siegel.
\newblock On universal classes of extremely random constant-time hash
  functions.
\newblock {\em SIAM Journal on Computing}, 33(3):505--543, 2004.

\bibitem{Sundar91}
Rajamani Sundar.
\newblock A lower bound for the dictionary problem under a hashing model.
\newblock In {\em Proceedings 32nd Annual Symposium of Foundations of Computer
  Science (FOCS)}, pages 612--621, 1991.
\newblock \href {https://doi.org/10.1109/SFCS.1991.185427}
  {\path{doi:10.1109/SFCS.1991.185427}}.

\bibitem{openproblems}
Mikkel Thorup.
\newblock Mihai {P}{\v a}tra{\c{s}}cu: Obituary and open problems.
\newblock {\em Bulletin of EATCS}, 1(109), 2013.

\end{thebibliography}
